\documentclass[11pt]{article}
\usepackage{amsmath,amssymb}
\usepackage{amsthm}
\usepackage{algorithm}
\usepackage{algpseudocode}
\usepackage{subcaption}
\usepackage{caption}
\usepackage[table]{xcolor}
\usepackage{xcolor}
\usepackage{tcolorbox}
\usepackage{tikz}
\usepackage{bbm}
\usepackage{wasysym}
\usepackage{fullpage}
\usepackage{verbatim}
\usepackage{dirtytalk}
\usepackage{graphicx}
\usepackage{tikz}
\usepackage{amsfonts}
\usepackage{hyperref}
\usepackage{mathtools}
\usepackage{csquotes}
\usepackage{listings}
\usepackage{thmtools}
\usepackage{cleveref}
\usepackage{tikz-cd}
\usepackage{minibox}
\usepackage{braket}
\usepackage{soul}
\usepackage{dsfont}
\usepackage{doi}

\makeatletter
\newcommand\ketbra[2][]{%
  \def\ketbra@firstarg{#1}%
  \def\ketbra@secondarg{#2}%
  \ifx\ketbra@firstarg\empty%
    \left\lvert\ketbra@secondarg\middle\rangle \! \middle\langle \ketbra@secondarg\right\rvert%
  \else%
    \left\lvert\ketbra@firstarg\middle\rangle \! \middle\langle\ketbra@secondarg\right\rvert%
  \fi%
}
\makeatother

\makeatletter
\def\metadef#1#2{%
  \def\metadef@iter##1{\ifx##1;\else \expandafter\newcommand\csname#1\endcsname{#2}\expandafter\metadef@iter\fi}%
  \expandafter\metadef@iter%
}
\makeatother

\metadef{vec#1}{\mathbf #1}abcdefghijklmnopqrstuvwxyz;

\metadef{ten#1}{\mathbf #1}ABCDEFGHIJKLMNOPQRSTUVWXYZ;

\metadef{mat#1}{\mathbf #1}ABCDEFGHIJKLMNOPQRSTUVWXYZ;

\metadef{frak#1}{\mathfrak #1}ABCDEFGHIJKLMNOPQRSTUVWXYZ;

\metadef{c#1}{\mathcal #1}ABCDEFGHIJKLMNOPQRSTUVWXYZ;

\metadef{s#1}{\mathsf #1}ABCDEFGHIJKLMNOPQRSTUVWXYZ;

\newtheorem{theorem}{Theorem}
\numberwithin{theorem}{section}
\newtheorem{claim}       [theorem] {Claim}
\newtheorem{lemma}       [theorem] {Lemma}

\newtheorem{definition}  [theorem] {Definition}

\newtheorem*{theorem*}{Theorem}

\DeclareMathOperator{\Var}{Var}

\newcommand{\eps}{\varepsilon}
\newcommand{\poly}{\mathsf{poly}}
\newcommand{\negl}{\mathsf{negl}}

\newcommand{\randfrom}{\xleftarrow{\$}}
\newcommand{\U}{\mathcal{U}}
\newcommand{\E}{\mathop{\mathbb{E}}}

\newcommand{\wt}[1]{\widetilde{#1}}
\newcommand{\floor}[1]{\left\lfloor#1\right\rfloor}
\newcommand{\abs}[1]{\left|#1\right|}

\renewcommand{\set}[1]{\left\{#1\right\}}

\def\sseq{\subseteq}

\DeclarePairedDelimiter\ceil{\lceil}{\rceil}

\newcommand{\N}{\mathbb{N}}

\newcommand{\cala}{{\mathcal A}}

\newcommand{\cald}{{\mathcal D}}

\newcommand{\calh}{{\mathcal H}}

\newcommand{\calo}{{\mathcal O}}

\newcommand{\calu}{{\mathcal U}}

\newcommand{\setbrac}[1]{\left\{#1\right\}}
\newcommand{\blt}{\setbrac{0,1}}

\newcommand{\bbn}{{\mathbb N}}

\makeatletter
\newcommand\ExpOwsg[4][]{%
  \def\ExpOwsg@firstarg{#1}%
  \ifx\ExpOwsg@firstarg\empty%
    \mathsf{Exp}_{#2,#3}(#4)
  \else%
    \mathsf{Exp}_{#2,#3,#1}(#4)
  \fi%
}
\makeatother

\def\Exp{\mathop{{}\mathbb{E}}}
\renewcommand{\P}{\mathsf{P}}
\newcommand{\NP}{\mathsf{NP}}
\newcommand{\PP}{\mathsf{PP}}
\newcommand{\BPP}{\mathsf{BPP}}
\newcommand{\SZK}{\mathsf{SZK}}
\newcommand{\BQP}{\mathsf{BQP}}
\newcommand{\QCMA}{\mathsf{QCMA}}
\newcommand{\QMA}{\mathsf{QMA}}
\newcommand{\coRQP}{\mathsf{coRQP}}

\newcommand{\A}{\mathcal{A}}

\newcommand{\Samp}{\mathsf{Samp}}

\newcommand{\Ver}{\mathsf{Ver}}
\newcommand{\OWP}{\mathsf{OWPuzz}}

\newcommand{\PRS}{\mathsf{PRS}}
\newcommand{\OWSG}{\mathsf{OWSG}}
\newcommand{\distOWP}{\mathsf{distOWPuzz}}

\newcommand{\K}{\mathsf{K}}
\newcommand{\KT}{\mathsf{KT}}
\newcommand{\Kt}{\mathsf{Kt}}
\newcommand{\Kkt}{\mathsf{K}^t}
\newcommand{\MCSP}{\mathsf{MCSP}}
\newcommand{\MQCSP}{\mathsf{MQCSP}}
\newcommand{\GapMQCSP}{\mathsf{GapMQCSP}}
\newcommand{\UMCSP}{\mathsf{UMCSP}}
\newcommand{\SMCSP}{\mathsf{SMCSP}}
\newcommand{\GapK}{\mathsf{GapK}}

\newcommand{\GapMCSP}{\mathsf{GapMCSP}}

\newcommand{\Kq}{\mathsf{Kq}}
\newcommand{\QC}{\mathsf{QC}}
\newcommand{\KH}{\mathsf{H}}
\newcommand{\GapKq}{\mathsf{GapKq}}
\newcommand{\GapQC}{\mathsf{GapQC}}
\newcommand{\GapH}{\mathsf{GapH}}

\newcommand{\flws}{\randfrom}

\newcommand{\SD}{\mathsf{SD}}

\newcommand{\D}{\mathcal{D}}

\usepackage{graphicx} 
\definecolor{corlinks}{RGB}{200,0,0}
\definecolor{cormenu}{RGB}{200,0,0}
\definecolor{corurl}{RGB}{200,0,0}

\hypersetup{
 colorlinks=true,
 urlcolor=corlinks,
 linkcolor=corlinks,
 menucolor=cormenu,
 citecolor=corlinks,
 pdfborder= 0 0 0
}

\title{A Meta-Complexity Characterization of Quantum Cryptography}

\author{
    Bruno P. Cavalar\thanks{Email: \texttt{bruno.cavalar@cs.ox.ac.uk}}
    \vspace{0.2cm}
    \\{\small University of Oxford\vspace{0.3cm}}
    \and
    Eli Goldin\thanks{Email: \texttt{eli.goldin@nyu.edu}}
    \vspace{0.2cm}
    \\{\small New York University\vspace{0.3cm}}
    \and
    Matthew Gray\thanks{Email: \texttt{matthew.gray@magd.ox.ac.uk}}
    \vspace{0.2cm}
    \\{\small University of Oxford\vspace{0.3cm}}
    \and
    Peter Hall\thanks{Email: \texttt{pf2184@nyu.edu}}
    \vspace{0.2cm}
    \\{\small New York University\vspace{0.3cm}}
}

\begin{document}

\maketitle

\begin{abstract}
    We prove the first meta-complexity characterization of a quantum
    cryptographic
    primitive. We show that 
    one-way puzzles
    exist if and only if there is some
    quantum samplable distribution 
    of binary strings
    over which it is hard to approximate
    Kolmogorov complexity.
    Therefore, we characterize one-way puzzles by the average-case hardness
    of a \emph{uncomputable} problem.
    This brings to the quantum setting
    a recent line of work that characterizes
    classical cryptography with the average-case hardness of a meta-complexity
    problem, initiated by Liu and Pass.
    Moreover, since the average-case hardness of Kolmogorov complexity
    over \emph{classically} polynomial-time samplable distributions
    characterizes one-way functions,
    this result poses one-way puzzles as a natural generalization 
    of one-way functions to the quantum setting.
    Furthermore, our equivalence 
    goes through probability
    estimation, giving us the additional equivalence that one-way puzzles
    exist if and
    only if there is a quantum samplable distribution over which
    probability estimation is hard. 
    We also
    observe that 
    the oracle worlds of defined by Kretschmer et. al.
    rule out any relativizing characterization
    of 
    one-way puzzles
    by the hardness of a problem in $\NP$ or $\QMA$,
    which means that 
    it may not be possible with current techniques
    to characterize one-way puzzles with another meta-complexity problem.
\end{abstract}

\pagebreak
\tableofcontents
\pagebreak
\section{Introduction}


What is the minimal complexity-theoretic assumption required for
\emph{quantum}
cryptography?
This is a fundamental question which even in the classical case began to be
understood only recently.
\emph{One-way functions (OWFs)}
are a basic, minimal primitive in classical cryptography,
being necessary for a wide variety of cryptographic tasks,
as well as being equivalent to many
others~\cite{HILL99,impagliazzo1989one,FOCS:GolGolMic84,impagliazzo1995personal}.
%
It is easy to observe that if one-way
functions exist 
then
$\P\neq \NP$ (the most fundamental complexity-theoretic conjecture).
Furthermore, there is a long tradition of building cryptographic primitives
from the computational hardness of \textit{specific} problems, e.g.
factoring, Learning-with-Errors, etc.~\cite{RivShaAdl78,C:BFKL93}. 
However, these
implications only go in a single direction. 

The ``holy grail'' result of this form~\cite{passholygrail} would
be to show that one-way functions exist if and only if $\P \neq \NP$. Note
that one-way functions are the ``central'' classic cryptographic primitive. 
As remarked above, one way
functions are minimal --- almost all classical modern cryptography implies
the existence of OWFs, mostly through trivial reductions. Moreover, one way
functions are useful --- they can be used to build everything in the
``crypto-complexity'' class Minicrypt \cite{impagliazzo1995personal} including
commitments, pseudorandom generators, and secret key cryptography.
Similarly, $\P \neq \NP$ is the central conjecture in complexity theory. If it
is shown to be false, much of complexity theory becomes trivial. Thus, if the holy grail result is true, then it
implies the more general claim that 
``classical cryptography is possible if and only
if complexity theory is interesting.''

For many years, researchers have tried and failed to achieve this holy
grail~\cite{FOCS:BogTre03,STOC:AGGM10} \cite{TCC:BogBrz15}. However, some progress
has come in recent years from a surprising 
front:
meta-complexity. 

\paragraph{Meta-complexity characterizations of cryptography.}

Meta-complexity 
refers to
computational problems which are themselves
about computational complexity. The foremost example of a meta-complexity
problem is the task of computing the Kolmogorov complexity $\K(x)$ of a string
$x$, which is defined as 
the length of the shortest Turing machine outputting 
$x$.
The theory of meta-complexity is known to have
numerous applications~\cite{DBLP:series/txcs/LiV19},
including some recent discoveries in learning algorithms, hardness
magnification, pseudorandomness and worst-case to average-case
reductions~\cite{DBLP:conf/coco/CarmosinoIKK16,STOC:Hirahara21,santhanam2020pseudorandomness,ITCS:CHOPRS20}.

In a recent breakthrough, Liu and Pass~\cite{liu2020one}
have shown that OWFs exist if and only the
if problem of
computing
$\Kkt$
(i.e., 
    the length of the shortest program that outputs $x$ in
$t(\abs{x})$ steps)
is
\emph{weakly hard on average} over the uniform distribution
for some polynomial $t$.
Later works have shown 
similar characterizations of OWFs
from the hardness of various other meta-complexity problems \cite{C:LiuPas23,DBLP:conf/stoc/LiuP21,santhanam2020pseudorandomness,ilango2021hardness,hirahara2023capturing,hirahara2023duality}.
%
%
%
These results reveal both
a deep connection
between cryptography and 
meta-complexity, 
and 
OWFs role as a central node equating many different meta-complexity tasks and properties.

Of particular interest to this work is~\cite{ilango2021hardness},
which shows that one-way functions exist
if and only if there exists a \emph{samplable distribution} on
which Kolmogorov complexity is hard to approximate.
This result has the advantage of applying not only to the uniform
distribution
but to any distribution that is samplable by a polynomial-time algorithm,
a more natural setting from the perspective of the theory of average-case
complexity~\cite{DBLP:journals/fttcs/BogdanovT06}.



\paragraph{The complications of the quantum setting.}

The above focus on 
OWFs as the minimal assumption for computational
cryptography has been complicated by recent work in quantum cryptography.
Starting with Ji, Liu, and Song's introduction of pseudorandom states (PRS)~\cite{ji2018pseudorandom},
recent works~\cite{C:MorYam24,EPRINT:MorYam22c,C:AnaQiaYue22} have defined a new
suite of quantum cryptographic primitives, most of which appeared to be
less powerful than OWFs. Kretschmer formalized this by introducing an
oracle world in which PRSs (and the many primitives which they imply) exist,
but $\BQP=\QCMA$ \cite{Kretschmer21} meaning that no post-quantum OWFs are
possible\footnote{Post-quantum OWFs are OWFs secure against quantum
    adversaries. The paper~\cite{KQST23} similarly shows a classical oracle
under which PRSs exist but $\P=\NP$.}. 

Thus, in the quantum setting, the classical holy grail result is no longer
relevant. It is possible that $\P = \NP$, but useful cryptography exists. This
brings up the question

\begin{center}
    \textit{What is the minimal natural complexity assumption required for quantum cryptography?}
\end{center}

To answer this question, 
it would be worthwhile to look at
the minimal primitive for quantum cryptography. 
However, even this is not clear in the quantum setting.
Though one-way functions play this role in classical 
cryptography,
and indeed all the cryptographic primitives of Minicrypt are known to be equivalent to it~\cite{impagliazzo1995personal},
the cryptographic primitives that follow from PRSs
(the collection of which some have called
\emph{Microcrypt}~\cite{morimaetalk})
are not known to be equivalent to each other, nor is there
an unequivocably minimal primitive underlying all others.

Currently the two weakest 
candidate
primitives
are quantum bit commitments (EFI)
\cite{brakerski2022computational,yan2022general} which seem to be minimal
for settings with quantum communication, and one way puzzles (OWPuzz) which
seem to be minimal in the quantum-computation classical-communication
(QCCC) \cite{STOC:KhuTom24,chung2024central} setting.
One-way puzzles can also be
constructed from a number of important quantum output primitives, such as
pseudorandom states and one-way state generators~\cite{STOC:KhuTom24,cavalar2023computation}. Notably, the existence of OWPuzz implies that 
$\BQP \neq \PP$~\cite{cavalar2023computation}. On the other hand, there are
barriers to showing that the existence of EFI implies any conjectures in
complexity theory~\cite{STOC:LomMaWri24}.


In this work, we will 
characterize one-way puzzles with the average-case hardness of estimating
Kolmogorov complexity.
Our results will naturally generalize the characterizations
of~\cite{ilango2021hardness},
corroborating
the centrality of one-way puzzles in
quantum cryptography,
and
providing the first complexity-theoretic characterization of
a quantum cryptographic primitive.
Furthermore, we will argue that this is in some sense
the only possible 
characterization of
one-way puzzles 
based on the hardness of a meta-complexity problem.

\subsection{A Characterization of One-way Puzzles}
\label{sec:results}

Our 
main 
result is an equivalence between the existence of one-way puzzles
and the average-case hardness of estimating Kolmorogov complexity over
quantum samplable distributions.
An interesting aspect of this result is that, though Kolmogorov complexity is known
to be \emph{uncomputable}~\cite{DBLP:series/txcs/LiV19},
it nevertheless
characterizes cryptography in the average-case setting.
 
Let $\GapK[s,t]$ be the promise problem of distinguishing between
binary strings of Kolmogorov complexity at most $s$
and those of Kolmogorov complexity at least $t$.
A problem is said to be 
\emph{weakly average-case hard}
on a distribution $\cald$ if,
for all sufficiently large $n$,
every 
quantum
polynomial-time algorithm
gives the wrong answer with probability at least $n^{-O(1)}$.

We now recall the definition of one-way puzzles~\cite{STOC:KhuTom24}. 
A one-way puzzle is defined by a
quantum sampler algorithm, which samples a puzzle and a corresponding key,
as well as an inefficient verification algorithm. The puzzle and key for a
OWPuzz must be classical strings. The puzzle should be easy to sample, but
hard to solve. That is, the sampled key-and-puzzle pairs should pass verification,
but given just the puzzle it should be hard to find a key which 
passes the test.
We show that, just like with OWFs, one-way puzzles can be characterized by the
average-case hardness of 
computing $\GapK$.

\begin{theorem}
    \label{t:main-result-owpuzz}
    The following are equivalent:
    \begin{enumerate}
        \item One way puzzles exist.
        \item 
        There exists a quantum samplable distribution $\D$ on which
        probability estimation (\Cref{def:probest}) is weakly hard on average.
        \item
            For some $s = n^{\Omega(1)}$ and $\Delta = \omega(\log n)$,
            there exists a quantum samplable distribution $\cald$ such that
            $\GapK[s,s+\Delta]$ is quantum weakly average-case hard on $\cald$.
    \end{enumerate}    
\end{theorem}

Interestingly, the only difference between our hard problem and the
hard problem of~\cite{ilango2021hardness} characterizing OWFs is that our
hardness is over quantum samplable distributions,
whereas theirs is over
classically samplable distributions. 
Indeed, 
one of their main results state that replacing every reference
to \emph{quantum}
by \emph{probabilistic} in Items 2 and 3 of our~\Cref{t:main-result-owpuzz}
yields a characterization of OWFs. Consequently, our result gives an exact characterization of Microcrypt from meta-complexity: we live in Microcrypt if and only if 1) all classical distributions are easy to estimate Kolmogorov complexity on, and 2) there are quantum distributions over which estimating Kolmogorov complexity is hard.

\paragraph*{On the central importance of one-way puzzles.}
One way puzzles were introduced by Khurana and Tomer
\cite{STOC:KhuTom24} as an intermediate primitive between Morimae
and Yamakawa's one way state generator (OWSG) \cite{morimae2022quantum} and
EFI pairs, and were studied as an object of independent interest by Chung,
Goldin, and Gray \cite{chung2024central}. They fall near the base of the
current hierarchy of quantum cryptographic primitives (only EFI is believed
to be a weaker primitive) and have several properties desirable for a
cryptographic primitive including combiners, a universal construction, and
hardness amplification. 

Prior work shows that a large number of quantum primitives (pseudorandom
states, one-way state generators, and all of QCCC cryptography) can be used
to build one-way puzzles~\cite{STOC:KhuTom24,chung2024central}. Furthermore,
one-way puzzles can be used to build EFI pairs, as well as everything else
equivalent to them, such as multi-party computation and quantum bit
commitments~\cite{STOC:KhuTom24}. 
However, 
we do not yet know of any efficient cryptographic primitive that can be built using
one-way puzzles
that cannot also be built using EFI pairs\footnote{Concurrent work shows that one-way puzzles can be used to build inefficiently-verifiable proofs of quantumness, an inefficient yet important primitive~\cite{morimae2024cryptographiccharacterizationquantumadvantage}}.
In fact,
there are barriers to using one-way puzzles to build a number of primitives
other than EFI pairs~\cite{chung2024central}. 
Thus, 
the place of one-way puzzles in quantum cryptography is not yet
well-understood.

One way to view our result is as giving more evidence that one-way puzzles are
valuable to study as an independent primitive. 
Indeed, 
since \Cref{t:main-result-owpuzz} is 
the natural quantum generalization of a characterization of one-way
functions,
our result can be interpreted as showing that one-way puzzles
are the natural quantum generalization of one-way functions.
Furthermore,
given the central role that the hardness of meta-complexity plays in classical
cryptography,
one can read this result
as saying that one-way puzzles embody some
fundamental aspect of quantum cryptographic hardness.




\paragraph*{Other meta-complexity characterizations are difficult.}
Recall that there exists an oracle relative to which OWPuzz exist but 
$\BQP = \QMA$. 
This means that any characterization of OWPuzz by a problem in $\QMA$ must
by necessity be non-relativizing. However, most meta-complexity
characterizations of OWFs are based on problems in $\NP$ \textit{do} relativize.

In particular, other meta-complexity problems known to characterize OWFs
can easily be seen to be in $\NP$ or $\QMA$. 
Examples include the (classical or quantum) minimum circuit size
problem as well as time-bounded Kolmogorov
complexity~\cite{liu2020one,chia2021quantum}.
Thus, the paper of Kretschmer~\cite{Kretschmer21} 
provides a concrete barrier to showing that
variants of these problems characterize one-way puzzles. 
Concretely, any such characterization must avoid the relativizing
techniques usually employed in meta-complexity, or use one of the handful of problems (such as $\K^{exp}$ or $\Kt$) which do not have efficiently checkable witnesses.

\paragraph*{An observation about the hardness of estimating Kolmogorov complexity with an $\NP$ oracle.}

Finally, since our results do relativize, and since there exists an oracle
relative to which $\P = \NP$ but OWPuzz exist, our results also provide a
barrier to showing that estimating Kolmogorov complexity on average (on non polynomial time samplable distributions) can be solved
using an $\NP$ oracle. This is not surprising, 
since computing Kolmogorov
complexity exactly is undecidable~\cite{DBLP:series/txcs/LiV19}, 
but this is not an
observation we are aware of existing elsewhere.

\subsection{Related Work}

\paragraph{Other Meta-Complexity Characterizations.}


Almost all the meta-complexity characterizations of cryptography have shown
equivalences between OWFs and the hardness of a specific meta-complexity
problem on some kind of distribution. Two results however fall outside this
paradigm and are of particular interest. 

First is the work of Hirahara et al \cite{hirahara2023duality} which showed that the existence of OWFs
is equivalent to the failure of symmetry of information to hold on average
for a probabilistic time-bounded notion of 
Kolmogorov complexity. 

Second is the work of Ball, Liu, Mazor, and Pass \cite{ball2023kolmogorov}, which first extended this
program outside Minicrypt. In their work they characterize the existence of key-agreement (generally
considered the 
domain
of ``Cryptomania'') by showing that the existence of
key-agreement (KA) protocols is equivalent to the worst case hardness of an interactive
promise notion of Kolmogorov complexity.

\paragraph{Quantum Meta-Complexity.}
In the aftermath of Bernstein and Vazirani's robust definition of quantum
Turing machines \cite{bernstein1997quantum} and their popularization
amongst complexity theorists by Fortnow \cite{fortnow2000one},
four main definitions $\Kq$ \cite{vitanyi2000three}, 
$\QC$ \cite{berthiaume2001quantum}, 
$\KH$ \cite{gacs2001quantum}, and 
$\K_{net}^\epsilon$ \cite{mora2004algorithmic} 
were proposed as extensions of Kolmogorov complexity that could measure the complexity of quantum states. 
Pseudorandom states can be seen to trivially imply the hardness of
estimating any of those measures on the quantum samplable distribution
consisting of either many copies of a Haar random state, or many copies of
a $\PRS$ output. But, besides this trivial implication, there have been no
results connecting them to quantum 
cryptography. When restricted to classical strings all four definitions are
equivalent to classical Kolmogorov complexity. Consequently our result
immediately gives corollary versions of our main theorem by replacing
$\GapK$ with $\GapKq$, $\GapQC$, $\GapH$, or $\GapK_{net}^\epsilon$, and
specifying that the distribution $\D$ is over classical states. 

Recently Chia, Chou, Zhang, and Zhang gave three extensions of the minimum
circuit size problem ($\MCSP$) to the quantum setting
\cite{chia2021quantum}: $\MQCSP$ which measures the quantum circuit
complexity of classical strings, $\UMCSP$ which measures the complexity of
implementing a unitary, and $\SMCSP$ which measures the complexity of
generating a quantum state. They show that several implications between
cryptography and the hardness of MCSP generalize to the quantum setting. In
particular they generalize the result of Allender and Das that $\SZK
\subseteq \BPP^{\MCSP}$ to $\SZK \subseteq \BPP^{\MQCSP}$
\cite{allender2017zero}, and the observation (which can be found in
\cite{ilango2021hardness}) 
that OWFs imply $\GapMCSP$ is hard to 
post-quantum OWFs imply $\GapMQCSP$ is hard. 
They show the interesting
implication that quantum-secure $i\mathcal{O}$ plus $\MQCSP \in \BQP$ imply
that $\NP \sseq \coRQP$, and make the observation that 
the existence of $\PRS$ implies that
computing $\SMCSP$, just like estimating the quantum Kolmogorov measures above, is hard.






\paragraph{Concurrent work.} A concurrent work by Khurana and Tomer independently proves that the existence of one-way puzzles is equivalent to the hardness of probability estimation on quantum samplable distributions~\cite{KT24b}, i.e. the first part of our~\Cref{t:main-result-owpuzz}. The reason they needed this lemma was for a different purpose than ours. While we use this to show that one-way puzzles are characterized by hardness of a meta-complexity problem, they are interested in building one-way puzzles from sampling assumptions. In addition, they use this lemma to show that one-way puzzles are equivalent to state puzzles, one-way puzzles where the key is quantum. 

Another concurrent work by Hiroka and Morimae independently shows one direction of our result, namely that the hardness of $\GapK$ implies the existence of one-way puzzles. Their proof also goes through probability estimation~\cite{concurrent}.

\subsection{Technical Overview}


\noindent We go over the main technical insights in proving the equivalence of (1) the existence of one way puzzles, (2) the hardness of probability estimation, and (3) the hardness of $\GapK$ estimagtion (as in Theorem~\ref{t:main-result-owpuzz}). We first prove that $\lnot (1) \Rightarrow \lnot (2)$, i.e., 
the nonexistence of one way puzzles implies that probability estimation is easy on all
quantum distributions. We then prove $\lnot (2) \Rightarrow \lnot (3)$, i.e., that probability
estimation being easy implies that $\GapK$ is easy on all distributions. Finally, we prove that $\lnot (3) \Rightarrow \lnot (1)$, i.e., that $\GapK$ being easy on all distributions is sufficient for breaking $\OWP$. Concluding the loop and showing that the three
statements in our theorem are equivalent.

\paragraph{Nonexistence of One-way Puzzles $\implies$ Probability Estimation.}

The task in this section is to estimate $p_y = \Pr[y \sim \D]$ using one-way puzzle inverters. In the works of Ilango, Ren, and Santhanam
\cite{ilango2021hardness} and Impagliazzo and Luby
\cite{impagliazzo1989one}, the parallel step of estimating $p_x$ using one-way function inverters is done by treating a classical
distribution $\D$ as a function $f$ mapping random inputs $r$ to samples $y
= f(r)$. They then use Valiant and Vazirani \cite{valiant1985np} style
hashing tricks which intuitively work as follows. 

Instead of only inverting $f$, they invert the function $f'(r, h) = f(r), h, h(r)$, where $h$ is the description of a hash function. By extending the length of the hash's output and fixing the bits of $z= h(r)$ at random (i.e., sampling $h, z$ and giving $f(r), h, z$ as input to our inverter), the space of valid preimages can be cut in half again and again until eventually only one valid preimage remains. At this point, the only preimage of $f'(r, h)$ will be $r$ itself, so further increases to the length will with high probability leave no valid preimages. If we try each possible output length, then, we will be able to approximately identify this point with very high probability allowing us to estimate $p_y$.

This technique fails to naively transfer to the quantum setting because
quantum distributions cannot be purified in the way that classically random
ones can. 
That is, classically samplable distributions can have
the randomness removed from the sampling process
and isolated as an input to the sampler, but quantum distributions cannot have
their ``quantum randomness'' removed in this way. However, two facts give us some hope. 

First, the proof of the equivalence between one-way functions and the hardness of probability estimation laid out above uses the exact same techniques as is used to prove one-way functions are equivalent to distributional one-way functions, a similar primitive where now only distributional inversion is hard. Consequently breaking a distributional one-way function requires a stronger inverter which can not only find a single valid preimages but can sample from the distribution of preimages. 

Second, the recent work of ~\cite{chung2024central} showed that we can get the quantum version of this result showing that one-way
puzzles exists if and only if \emph{distributional one-way puzzles} (whose security guarantee ) do. Meaning that the non existence of $\OWP$ gives us the powerful tool of distribution inverters. Using these, and hashing over the outputs instead of the now inaccessible random inputs we are able to complete our task as follows.

Given a quantum distribution $\D$, we define a
distributional one-way puzzle
treating $y$ as a key and $h,h(y)$ as the puzzle. Using similar
intuition as above we can think of each bit $h(y)_i$ as cutting the
distribution of other valid outputs $x \neq y$ in half. When $t = 0$ every
output is a valid inversion. When $t=1$ then we would expect only a subset of outputs with
overall probability of one half are likely to be considered valid. And when
$t = -\log(\Pr_D[y])$ we'd expect only a $2^{-t}$ fraction be considered
valid. 
This intuition holds as long as a noticeable fraction of outputs occur with probability $2^{-t}$. 

By analyzing the distribution over the probability mass of the valid
outputs, we are able to show that distributionally inverting such puzzles
allows you to reliably estimate the probability of an output up to a
polynomial multiplicative factor. This is somewhat weaker than what can be
obtained in the classical setting (i.e. a constant estimation factor).
However this is sufficient to distinguish between elements with a
probability gap of $\omega(\log(n))$ as needed in the next section.

\paragraph{Probability Estimation $\implies$ Computing GapK.}

The reduction is from $\GapK$ to probability estimation is simple: We assume that strings with probability above the
midpoint in the gap  ($x$ s.t. $\Pr[x] > 2^{-m}$) have ``low'' Kolmogorov complexity, and that strings
with probability below the midpoint $\Pr[x] < 2^{-m}$ have ``high'' Kolmogorov complexity.
For this reduction to be valid, we need to prove that this procedure makes
few mistakes in each direction. 

First, we show that strings output from a quantum distribution with high probability but high Kolmogorov complexity do not exist. To do this, we prove a quantum
generalization of the coding theorem which shows how to create a short description of any high probability string. This proof is based on the
observation that a time unbounded machine (like those considered in
Kolmogorov complexity) can simulate the running of a quantum machine. Next we show that we make few errors caused by strings
with low probability and low Kolmogorov complexity. We argue that such
strings must be sampled by the distribution with only negligible
probability. This is because there are few strings with low Kolmogorov
complexity, and by definition they individually have low probability of
being sampled.

\paragraph{Computing GapK $\implies$ Nonexistence of One-way Puzzles.}
This section diverges the farthest from the classical
equivalent~\cite{ilango2021hardness}, as they use the fact that OWFs imply
PRGs with arbitrary stretch~\cite{HILL99,VZ12}. The outputs of such PRGs
have very low Kolmogorov complexity (and $\K^t$ complexity) and so can be
distinguished straightforwardly from random in the PRG security game.
Therefore $\GapK$ being easy implies that there are no PRGs and therefore
no OWFs. However a quantum equivalent to \cite{HILL99} is not known, and
neither $\OWP$ nor even $\OWSG$ are known to imply a quantum pseudo-random
sampler. 

Instead, we will use $\GapK$ to break 
what one could call a
\emph{non-uniform pseudorandom low-entropy distribution},
an intermediary construction used to build 
EFI pairs from $\OWP$ in 
a recent work by 
Chung, Goldin, and Gray~\cite{chung2024central}.
They are created by taking the machinery for building PRGs from OWFs
\cite{HILL99,VZ12} and applying them to $\OWP$s resulting
in something as close to a PRG as we can currently get from $\OWP$s. The
resulting primitive takes some small ($\log n$ bit) piece of
non-uniform advice (which corresponds to the total next-bit pseudoentropy
of the $\OWP$'s sampler) and then samples from a distribution which is
indistinguishable from uniform while having less than $n - n^{\Omega(1)}$ 
bits of
true entropy.

Unfortunately, this distribution is non-uniform and our argument only gives us $\GapK$ estimation on uniform distributions from the non-existence of one-way puzzles. 
However, because the advice strings needed are only $\log n$ bits long, a
uniform distribution which 
samples an advice string uniformly at random
will include a 
$1/\poly(n)$
fraction of the distribution that corresponds
to the correct advice we want to do well on. We are guaranteed to be able
to estimate $\GapK$ on this distribution and therefore on the
sub-distribution of interest. If $\GapK$ is easy on that sub-distribution
then we can invert the underlying $\OWP$.

\subsection{Open Problems}


We conclude with a few interesting open problems left by our work.

\begin{enumerate}
    \item \textbf{Other Characterizations and Implications for OWPuzz.} 
        The Kretschmer's black-box barrier for quantum cryptography~\cite{Kretschmer21} rules out any (relativizing)
        characterization of one-way puzzles by a meta complexity property 
        which can be computed by $\QMA$ or $\NP$ algorithms.
        Specifically, 
        that barrier rules out a relativizing proof that one-way puzzles
        imply the
        hardness of $\MCSP$, $\MQCSP$, $\K^{\poly}$, or $\KT$,
        as each of these problems
        have
        efficiently checkable witnesses. 
        Can we nonetheless still show an implication in the other direction
        by
        building one-way puzzles from their average-case hardness over quantum
        distributions?
        Moreover, it remains possible that OWPuzz can be
        characterized by the hardness of meta-complexity measures which
        allow for descriptions that take super-polynomial time to output
        the string in question such as $\Kt$.
    \item \textbf{Characterizing OWSGs.} One-way state generators (OWSGs) are a classical-input
        quantum-output generalization of one way functions. 
        Can their
        existence be characterized by the hardness of some 
        \emph{quantum}
        meta-complexity notion? 
        Such a characterization would have to overcome several barriers,
        such as likely needing to show an equivalence between OWSGs and
        distributional OWSGs (generalizing a result of Cao and Xue
        \cite{cao2022constructing} which is limited to their symmetric setting).
    \item \textbf{Meta-complexity and EFIs.} EFIs are pairs of computationally
        indistinguishable and statistically far mixed states which are
        equivalent to quantum bit commitments. 
        This work implies that the
        average-case hardness of $\GapK$ on quantum distributions implies
        EFIs. However, it is unclear whether the existence 
        of EFIs has any direct meta-complexity implications. In fact, there exist weak barriers to showing that the existence of EFIs has any complexity theory implications at all~\cite{STOC:LomMaWri24}. Nevertheless, perhaps meta-complexity (or quantum meta-complexity) holds the key to bypassing this barrier.
    \item \textbf{Strong average-case hardness of Kolmogorov complexity.}
        One of the strengths of the results of~\cite{ilango2021hardness}
        compared to previous ones is that their characterization of one-way
        functions not only holds with respect to the \emph{weak}
        average-case hardness of computing $\GapK$,
        but also with respect to \emph{strong} average-case hardness.
        Moreover, their result is robust with regards to 
        any
        choice of parameters 
        $s = n^\eps$ and $\Delta = n-s-\omega(\log n)$.
        Essential to that weak-to-strong derivation and robustness
        is the equivalence
        between
        pseudorandom generators and one-way functions~\cite{HILL99}.
        Unfortunately, because it is currently only known that one-way
        puzzles imply
        quantum samplable distributions with linear pseudoentropy, we are
        only able to prove an equivalence in the weak average-case
        setting. 
        Can our results be extended to the strong average-case setting,
        and can there be more latitude in the choice of parameters?
\end{enumerate}

\section{Definitions and Preliminaries}

Throughout this paper we assume $\Delta$ and $s$ are polynomial time
computable functions from $\N \rightarrow \N$. 
By $\negl(n)$ we denote a
negligible function vanished by all polynomials: 
$\negl(n) = 1/n^{\omega(1)}$. 
We will write QPT as a shorthand for ``quantum polynomial time''.

\subsection{Probability}

Given two distributions $X,Y$,
we denote their statistical distance as 
$\SD(X ; Y)$ 
or 
$\SD(X,Y)$ which is equal to $\frac{1}{2}\sum_{z \in X \cup Y}(|\Pr[X \rightarrow z]-\Pr[Y \rightarrow z]|))$.
We denote by $\calu_n$ the uniform distribution over $\blt^n$.
Given a distribution $\cald$, we denote by $\set{\cald \to x}$ the event
that a sample of $\cald$ is equal to $x$.
We will also write
$x \flws \cald$ to denote that $x$ is sampled from $\cald$.
We denote by $H(\cald)$ the entropy of a probability distribution
$\cald$,
which is defined as
$H(\cald) := \Exp_{x \flws \cald}\left[-\log(\Pr[\cald \to x])\right]$.

\subsection{Cryptography and one-way Puzzles}

We say that a distribution $\cald$ supported on $\blt^m$
is computationally indistinguishable from the uniform distribution
if, for every $c \geq 1$ and every QPT algorithm $\cala$,
we have
\begin{equation*}
    \abs{
        \cala(\cald)
        -
        \cala(\calu_m)
    }
    <
    n^{-c},
\end{equation*}
for every large enough $n$.

\begin{definition}[$\OWP$]\label{def:owpuzz}
An $(\alpha,\beta)$ one way puzzle ($\OWP$) is a pair of a sampling algorithm and a verification function $(\Samp,\Ver)$ with the following syntax:
    \begin{enumerate}
        \item $\Samp(1^{\lambda}) \to (k,s)$ is a uniform QPT algorithm which outputs a pair of classical strings $(k,s)$. We refer to $s$ as the puzzle and $k$ as the key. Without loss of generality, we can assume $k \in \{0,1\}^\lambda$.
        \item $\Ver(k,s) \to b$ is some (possibly uncomputable) function which takes in a key and puzzle and outputs a bit $b \in \{0,1\}$.
    \end{enumerate}
    satisfying the following properties:\\
    \begin{enumerate}
        \item Correctness: For all sufficiently large $\lambda$, outputs of the sampler pass verification with overwhelming probability
        $$\Pr_{\Samp(1^{\lambda})\to (k,s)}[\Ver(k,s)\to 1] \geq 1 - \alpha$$
        \item Security: Given a puzzle $s$, it is computationally infeasible to find a key $s$ which verifies. That is, for all non-uniform QPT algorithms $\A$, for all sufficiently large $\lambda$,
        $$\Pr_{\Samp(1^{\lambda}) \to (k,s)}[\Ver(\A(s),s) \to 1] \leq \beta$$
    \end{enumerate}
    If for all $c$, $(\Samp,\Ver)$ is a $(\lambda^{-c},\lambda^{-c})$ one way puzzle, then we say that $(\Samp,\Ver)$ is a strong $\OWP$ and omit the constants. When unambigious, we will simply say that such a $(\Samp,\Ver)$ is a $\OWP$.
\end{definition}

\begin{definition}[$\distOWP$]
    \label{def:dist-owpuzz}
    Let
    $\Samp$
    be a QPT algorithm such that
    $\Samp(1^\lambda) \to (k,s)$ 
    is a pair of classical strings
    referred to as the key and the puzzle, respectively.
    We say that
    the sampling algorithm $\Samp$ is a 
    \emph{$\gamma$-secure distributional one-way puzzle}
    if, 
    %
    for every QPT algorithm $\cala$,
    we have
    \begin{equation*}
        \SD( (k,s) \;; (\cala(s), s))
        > \gamma,
    \end{equation*}
    where $(k,s)$ is sampled from 
    $\Samp(1^\lambda)$.
\end{definition}

%

\begin{theorem}[\protect{\cite[Theorem 33]{chung2024central}}]
    \label{thm:dist-owpuzz}
    If there exists a $n^{-c}$-secure distributional one-way puzzle for
    some $c \geq 1$,
    then there exists a one-way puzzle.
\end{theorem}

\subsection{Kolmogorov Complexity}


This section introduces Kolmogorov complexity and a computational problem
about estimating its value.
We refer the reader to~\cite{DBLP:series/txcs/LiV19} for properties about Kolmogorov
complexity.

\begin{definition} 
    Let $U$ be a universal Turing machine.
    For strings $x \in \{0,1\}^*$, 
    the Kolmogorov complexity $\K_U(x)$ of $x$ is the length of
    the shortest program $\rho$ such that $U(\rho)$ will halt and output $x$
    after a finite number of steps.
\end{definition}

Our results hold for every universal Turing Machine $U$, so
we will omit $U$ and simply write $\K_U$ as $\K$.  

\begin{definition}
    \label{def:GapK}
    For two functions 
    $s_1(n), s_2(n) : \bbn \to \bbn$
    satisfying,
    $0 < s_1(n) < s_2(n) < n$,
    we denote by
    $\GapK[s_1,s_2]$ the promise problem where YES instances are strings
    with $\K(x) \leq s_1(n)$ and NO instances are strings with $\K(x) \geq
    s_2$.
\end{definition}

\begin{lemma}
    The number of strings $x \in \blt^n$
    such that $\K(x) \leq t$ 
    is at most 
    $2^{t+1}-1$.
\end{lemma}

\subsection{Probability estimation}

We say that a quantum algorithm $\cald$
is a \emph{family of quantum samplable distributions}
if
$\cald$ is a QPT algorithm such that
$\cald(1^n) \in \blt^n$. We will use the shorthand 
$\cald_n := \cald(1^n)$.

We now define the task of probability estimation formally.
\begin{definition}
    \label{def:probest}
    Let $\cald = \set{\cald_n}_{n \in \bbn}$ be a 
    family of quantum samplable distribution.
    For $x \in \blt^n$, let $p_x$ be the
    probability that $x$ is sampled from $\cald_n$.
    We say that an algorithm $\cala$ performs
    \emph{probability estimation on $\cald$ infinitely often with error at
    most $\eps$ and precision $\delta$}
    if, for infinitely many $n \in \bbn$,
    we have
    $$
    \Pr_{\substack{x\randfrom \mathcal{D}_n \\ \cala}} [\delta \cdot p_x \leq \cala(x) \leq p_x] \geq 1 - \eps.
    $$
    We say that 
    \emph{probability estimation}
    on a distribution $\cald$
    is \emph{weakly hard on average}
    if there exists a constant $q \geq 1$
    such that no QPT algorithm can 
    perform probability estimation on $\cald$ infinitely often with error
    at most $n^{-q}$ and precision at least $n^{-O(q)}$.
\end{definition}

\subsection{Average-case complexity}

A decision or promise problem is said to
be
\emph{weakly average-case hard for quantum algorithms on a distribution $\cald$} if
there exists a constant $q \geq 1$
such that 
every 
quantum
polynomial-time algorithm
gives the wrong answer with probability at least $n^{-q}$
on inputs sampled from $\cald$,
for all sufficiently large $n$.
The problem is simply said to be
\emph{weakly hard on average for quantum algorithms}
if it's hard on some quantum samplable distribution.

An algorithm is said to \emph{solve} a problem \emph{with error at most
$\eps$ on a distribution $\cald$}
if
it gives the right answer with probability at least $1-\eps$
on inputs sampled from $\cald$.
A problem is said to be \emph{weakly easy on average} for quantum
algorithms on a distribution $\cald$ if,
for every constant $q \geq 1$, there exists a QPT algorithm that solves it
on $\cald$ with error at most $n^{-q}$
on infinitely many input lengths.

\section{Probability Estimation from the Nonexistence of One-way Puzzles}

In this section we prove
that the non-existence
of one-way puzzles implies algorithms for probability estimation on all
quantum samplable distributions.

\begin{theorem}[Item 2 $\implies$ Item 1]
    \label{thm:owpuzz-probest}
    If one-way puzzles do not exist, 
    then, for any family of quantum
    samplable distributions $\D$,
    probability estimation on
    the distribution $\cald$ is not weakly hard on average.
\end{theorem}

The result will follow from the following lemma.
\begin{lemma}
    \label{lemma:probest}

    Let $c > 0$.
    If there do not exist $n^{-c}$-secure distributional one-way puzzles,
    then, for any family of quantum samplable distributions $\D$,
    there is a QPT algorithm 
    $\cala$
    which multiplicatively estimates 
    $p_x\coloneqq \Pr[\D_n \to x]$ infinitely often.
    In other words,
    we have
    $$\Pr_{\substack{x\randfrom \mathcal{D}_n \\ \cala}}
    [p_x\leq \cala(x) \leq 4n^{2c}p_x] 
    \geq 1 - n^{1-c/2}$$
    for infinitely many $n \in \bbn$.

%
\end{lemma}

Indeed, \Cref{thm:dist-owpuzz} states that, if one-way puzzles do not exist,
then $n^{-c}$-secure distributional one-way puzzles do not exist
for every $c > 0$.
In particular, setting $c=2(d+1)$ in \Cref{lemma:probest}, 
\Cref{thm:owpuzz-probest} follows.
We now prove the lemma.

\begin{proof}[Proof of \Cref{lemma:probest}]
Let 
$$S = \set{x \in \blt^n : \Pr_{x' \flws \cald}[p_{x'} \leq p_x] \leq n^{-c}}.$$
We observe that we rarely sample elements from $S$, 
so we only need to succeed
in the probability estimation of $p_x$ for $x \notin S$.
\begin{claim}
    We have
    $\Pr_{x \flws \cald}[x \in S] \leq n^{-c}$.
\end{claim}
\begin{proof}
Note that, if $x \in S$ and $p_{x'} \leq p_x$, then $x' \in S$.
This implies that $x \in S$ iff $p_x \leq \theta$,
where
$\theta = \max\set{p_x : x \in S}$.
Therefore, we have
$\Pr[x \in S] = \Pr[p_x \leq \theta] \leq n^{-c}$.
\end{proof}


Let $\mathcal{H}_n^k$ be a family of 
3-wise
independent hash functions
with input $\{0,1\}^n$ and output $\{0,1\}^k$. 
On input $1^n$,
define a
distributional
one-way puzzle candidate $\Samp$ 
as follows:
\begin{algorithm}[H]
\caption{One-way puzzle candidate $\Samp$}
\begin{algorithmic}[1]
    \State \textbf{Input:} Number $n \in \bbn$ in unary ($1^n$).
    \State Sample $k \randfrom [2n]$.
    \State Sample $h\randfrom \mathcal{H}_n^k$.
    \State Sample $x \randfrom \D_{n}$.
    \State Output $(k,h,h(x))$ as the puzzle and $x$ as the key.
\end{algorithmic}
\end{algorithm}

Since $n^{-c}$-distributional one-way puzzles do not exist, there exists an
algorithm $\mathcal{O}$ which distributionally inverts $\Samp$ infinitely
often. In particular, for infinitely many $n$,
the distributions
\begin{align*}
    P_n &\coloneqq (k,h,h(x),x), \; \text{and}\\
\wt{P}_n &\coloneqq (k,h,h(x),\mathcal{O}(k,h,h(x)))
\end{align*}
satisfy
$$\Delta(P_n,\wt{P}_n) \leq n^{-c}.$$
From now on, fix $n \in \bbn$
such that the distributional inverter $\calo$ satisfies the above
inequality.
We now define the distribution $P_{k,x}$ as follows:
\begin{enumerate}
    \item Sample $h:\{0,1\}^n \to \{0,1\}^k$ a random 
        3-wise
        independent
        hash function from $\calh_n^k$.
    \item Sample $x' \randfrom \D_n$ conditioned on $h(x') = h(x)$.
    \item Output $x'$.
\end{enumerate}
We first observe that, if $k$ is large, then
the probability that $P_{k,x}$ outputs $x$ is large.
Moreover, if $k$ is small,
the probability that $P_{k,x}$ outputs $x$ is smaller.
\begin{lemma}
    \label{claim:approx}
    Let $m$ be such that $p_x \leq 2^{-m}$ and 
    $x \notin S$.

    If $k \leq m - 2c \log n-2$, then
    $$\Pr[P_{k,x} \to x] \leq 17n^{-c}$$
    If $k \geq m$, then
    $$\Pr[P_{k,x} \to x] \geq \frac{9}{10}$$
\end{lemma}

We define another distribution $\wt{P}_{k,x}$ as follows:
\begin{enumerate}
    \item Sample $h:\{0,1\}^n \to \{0,1\}^k$ a random 
        3-wise
        independent
        hash function from $\calh_n^k$.
    \item Output $\mathcal{O}(k,h,h(x))$.
\end{enumerate}
We claim that this distribution approximates $P_{k,x}$ with large
probability. 
\begin{lemma}
    \label{claim:close}
    Let $d \leq \frac{c-1}{2}$. 
    With probability at least $1-n^{-d}$ over $x$, for all $k \in [2n]$,
    $$\Delta(P_{k,x},\wt{P_{k,x}}) \leq n^{-d}.$$
\end{lemma}
Crucially, the claim shows that, with large probability over
$x \flws \cald$,
the distributions $P_{k,x}$ and $\wt{P_{k,x}}$ have small statistical
distance
\emph{for all $k$}.
Moreover, note that $P_{k,x}$ is not samplable, while $\wt{P}_{k,x}$ is. 
Together with \Cref{claim:approx},
this means we can approximate $p_x$
by 
sampling $\wt{P}_{k,x}$ over many values of $k$
and estimating its probability.
We formalize this idea in the following algorithm to approximate $p_x$.
Let $t = \poly(n)$ be a large enough polynomial.

\begin{algorithm}[H]
\caption{Probability estimation algorithm $\cala$}
    \label{alg:est}
\begin{algorithmic}[1]
    \For{$k \in [2n]$}
        \For{$j \in [t]$}
            \State
            Sample $h \flws \calh_n^k$
            \State
            Let $x'_{k,j} \leftarrow \calo(k,h,h(x))$
        \EndFor
        \State
        Let $c(k) \leftarrow $ the number of $j$ such that $x'_{k,j} = x$
    \EndFor
    \State
    Let $k^*$ be the smallest $k$ such that $c(k) \geq \frac{3}{8}t$. 
    If none exists, set $k^* = 2n$.
    \label{line:test}
    \State
    Output $2^{-(k^* - 1)}$.
\end{algorithmic}
\end{algorithm}


We claim that the algorithm correctly estimates $p_x$ when
$x \notin S$
and $x$ is such that $P_{k,x}$ and $\wt{P_{k,x}}$ are
statistically close.
Henceforth, let $d=\frac{c-1}{2}$. 
\begin{lemma}\label{claim:correctness}
    Let $x$ be such that $\Delta(P_{k,x},\wt{P}_{k,x}) \leq n^{-d}$ 
    for all $k$ and such that 
    $x \notin S$.
    We claim that 
    $$\Pr[p_x \leq \cala(x) \leq 4n^{2c}p_x] \geq 1 - n^{-c}$$
    infinitely often.
\end{lemma}
We are now ready to conclude the proof.
By~\Cref{claim:close}, we have
$\Pr_{x\randfrom \D_n}[\Delta(P_{k,x},\wt{P}_{k,x})
\geq n^{-d}] \leq n^{-d}$. 
Putting this all together, we get by \Cref{claim:correctness}
and the Claim
that
\begin{equation*}
    \Pr_{x \flws \cald_n}[p_x \leq \cala(x) \leq 2n^{2c}p_x] \geq 1 - 2n^{-c} - n^{-d} \geq 1 -
    3n^{-d} \geq 1 - n^{1-c/2}.
\end{equation*}
\end{proof}

\subsection{Proofs of claims}

\subsubsection{Proof of~\Cref{claim:approx}}

We first prove \Cref{claim:approx}. We will need the following lemmas.


\begin{lemma}\label{lem:firstlem}
    Let $m = \ceil{-\log p_x}$ and $\Pr_{\D \to x'}[p_{x'} \leq 2^{-m}] \geq n^{-c}$.
    Let $k \leq m - 2\log n^c$. 
    Then
    $$\Pr_{h}[\Pr_{\D' \to x}[h(x') = h(x)\text{ and }x'\neq x] \leq 2^{-k-2} n^{-c}] \leq 16n^{-c}$$
    where $h$ is drawn from a $3$-wise independent hash family $\{0,1\}^n \to \{0,1\}^k$.
\end{lemma}

\begin{proof}
    Note that if $\frac{1}{2^m} > n^{-c}/2$, then $k < 0$ so this trivially holds.

    Define $\varepsilon = \Pr_{\D \to x'}[p_{x'} \leq 2^{-m}\text{ and }x' \neq x] \geq n^{-c} - \frac{1}{2^m} \geq n^{-c}/2$.
    Define $\gamma = \Pr_{\D \to x'}[h(x') = h(x)\text{ and }x'\neq x\text{ and }p_{x'} \leq 2^{-m}]$, a random variable in $h$. Define $R_{x'} = p_{x'} \mathds{1}[h(x') = h(x)]$. We have
    $$\gamma = \sum_{x' \neq x:p_{x'} \leq 2^{-m}} R_{x'}.$$

    Note that $\E[R_{x'}] = p_{x'}\cdot \Pr[h(x') = h(x)] \leq p_{x'}\cdot 2^{-k}$.
    Furthermore, $\E[R_{x'}^2] = p_{x'}^2 \cdot \Pr_h[h(x') = h(x)] = p_{x'}^2 \cdot 2^{-k}$.
    And so 
    $$\E[\gamma] = \sum_{x' \neq x:p_{x'} \leq 2^{-m}}p_{x'}\cdot
    \Pr_h[h(x') = h(x)] = 2^{-k} \cdot \Pr_{\D \to x'}[p_{x'} \leq 2^{-m}]
    = 2^{-k}\cdot \varepsilon.$$

    By $3$-wise independence of $h$, the $R_{x'}$'s are pairwise independent, so
    $$\Var(\gamma) = \sum_{x' \neq x:p_{x'} \leq 2^{-m}} p_{x'}^2\cdot
    2^{-k} \leq \sum_{x' \neq x:p_{x'} \leq 2^{-m}} p_{x'}\cdot 2^{-(m+k)}
    = 2^{-(m+k)}\cdot \varepsilon.$$
    Chebyshev inequality then says that
    \begin{equation*}
        \begin{split}
            \Pr\left[\gamma \leq 2^{-k}\varepsilon - 2^{-k}\varepsilon/2\right] \leq \frac{\Var(\gamma)}{(2^{-k}\varepsilon/2)^2}\\
            = \frac{4}{\varepsilon}2^{-m+k}.
        \end{split}
    \end{equation*}
    And so plugging in $k = m - 2c\log n$ and $\varepsilon \geq n^{-c}$, we get
    $$\Pr\left[\gamma \leq 2^{-k-2} n^{-c}\right] \leq 16n^c \cdot
    2^{-2c\log n} \leq 16n^{-c}.$$
    Since $\Pr_{\D' \to x}[h(x') = h(x)\text{ and }x'\neq x] \geq \gamma$,
    the lemma follows.
\end{proof}

\begin{lemma}\label{lem:secondlem}
    Let $m = \ceil{-\log p_x}$.
    Then
    $$\Pr_{h}[\Pr_{\D' \to x}[h(x') = h(x)\text{ and }x'\neq x] \geq t\cdot 2^{-k}] \leq t^{-1},$$
    where $h$ is drawn from a $3$-wise independent hash family $\{0,1\}^n \to \{0,1\}^k$.
\end{lemma}

\begin{proof}
    Let $\alpha = \Pr_{\D' \to x}[h(x') = h(x)\text{ and }x'\neq x]$.
    A simple calculation of expectation gives us
    \begin{equation*}
        \begin{split}
            \E_h [\alpha] = \sum_{x' \neq x} p_{x'}\cdot 2^{-k}\\
            \leq 2^{-k},
        \end{split}
    \end{equation*}
    and so Markov bound says
    $$\Pr[\alpha \geq t\cdot 2^{-k}] \leq t^{-1},$$
    and we are done.
\end{proof}

We are now ready to prove~\Cref{claim:approx}.

%

\begin{proof}[Proof of \Cref{claim:approx}]
    Define $\alpha = \Pr_{\D' \to x}[h(x') = h(x)\text{ and }x'\neq x]$ a random variable in $h$.

    Let us first consider $k\leq m-2c\log n-2$.~\Cref{lem:firstlem} gives us
    $$\Pr[\alpha \leq n^{-c}\cdot 2^{-k-2}] \leq 16n^{-c}.$$
    Let us define $P_{k,h,x}$ to be $P_{k,x}$ conditioned on the hash function being $h$. Note that if $\alpha \geq n^{-c}\cdot 2^{-k-2}$, then
    \begin{align*}
        \Pr[P_{k,h,x} \to x] 
        &= \frac{p_x}{\alpha + p_x}
        \\&\leq \frac{p_x}{n^{-c}\cdot 2^{-k-2} + p_x}
        \\&\leq \frac{p_x}{n^{c}\cdot 2^{-m} + p_x}
        \\&\leq \frac{p_x}{n^{c}\cdot p_x + p_x}\leq \frac{1}{n^c + 1}
        \\&\leq n^{-c}.
    \end{align*}
    And so we get that
    $$\Pr[P_{k,x} \to x] \leq \Pr[\alpha \leq n^{-c}\cdot 2^{-k-2}] + n^{-c} = 17n^{-c}.$$

    We now consider the case where $k \geq m + 11$. By~\Cref{lem:secondlem}
    $$\Pr_{h}[\alpha \geq t\cdot 2^{-k}] \leq t^{-1}.$$
    Furthermore, if $\alpha \leq t\cdot 2^{-k}$, then
    \begin{equation*}
        \begin{split}
            \Pr[P_{k,h,x} \to x] = \frac{p_x}{\alpha + p_x}\\
            \geq \frac{p_x}{t\cdot 2^{-k} + p_x}\\
            \geq \frac{p_x}{t\cdot 2^{-10} p_x + p_x}\\
            \geq \frac{1}{t\cdot 2^{-10} + 1}
            \geq \frac{1}{t/1000 + 1}.
        \end{split}
    \end{equation*}
    Picking $t = 100$, we get 
    \begin{equation*}
        \begin{split}
            \Pr\left[P_{k,x} \to x] \geq \Pr[P_{k,h,x} \to x | \Pr_{\D' \to x}[\alpha \leq t\cdot 2^{-k}]\right]\Pr_{\D' \to x}[\alpha \leq t\cdot 2^{-k}]\\
            \geq \frac{99}{100}\frac{1}{1/10 + 1} \geq
            \frac{99}{100}\frac{10}{11} = \frac{9}{10}.
        \end{split}
    \end{equation*}
\end{proof}

\subsubsection{Proofs of \Cref{claim:close,claim:correctness}}

\begin{proof}[Proof of~\Cref{claim:close}]
    To ease notation, let $P = P_{k,x}$ and $\wt{P} = \wt{P_{k,x}}$.
    First, we will observe an equivalent sampling procedure for
    $P$.
    \begin{enumerate}
        \item Sample $k \randfrom [2n]$.
        \item Sample $h\randfrom \mathcal{H}_n^k$.
        \item Sample $x \randfrom \D_{n}$.
        \item Sample $x' \randfrom \D_n$ conditioned on $h(x') = h(x)$.
        \item Output $(k,h,h(x),x')$.
    \end{enumerate}
    We obtain
    \begin{align*}
            \Delta&(P,\wt{P}) = \frac{1}{2}\sum_{k,h,y,x'}
            \abs{\Pr[P \to (k,h,y,x')] - \Pr[\wt{P} \to (k,h,y,x')]}
            \\&=\frac{1}{n}\sum_k \left(\frac{1}{2}\sum_{h,y,x'}
                \abs{
                \begin{aligned}
                    \frac{1}{2^{\abs{h}}}&\Pr_{x \randfrom \D}[h(x) = y]\Pr_{x \randfrom \D}[x = x' | h(x) = y]
                    \\-& 
                    \frac{1}{2^{\abs{h}}}\Pr_{x \randfrom \D}[h(x) = y]\Pr[\mathcal{O}(k,h,y) = x']
                \end{aligned}
                }
        \right)
            \\&=\E_k\left[\sum_{h,y}\left(\Pr_{x\randfrom \D}[h(x) = y]\frac{1}{2}\sum_{x'}\abs{\frac{1}{2^{|h|}}\Pr_{x\randfrom \D}[x = x' | h(x) = y] - \frac{1}{2^{|h|}}\Pr[\mathcal{O}(k,h,y) = x']}\right)\right]\\
            \\&=\E_k\left[\sum_h\E_{y\randfrom
            h(\mathcal{D})}\left[\frac{1}{2}\sum_{x'}\abs{\frac{1}{2^{|h|}}\Pr_{x\randfrom
        \D}[x = x' | h(x) = y] - \frac{1}{2^{|h|}}\Pr[\mathcal{O}(k,h,y) =
        x']}\right]\right].
    \end{align*}
    We briefly explain each of the equalities above.
    The first follows the definitions of $P$ and $\wt{P}$,
    using also the fact that $k$ is sampled uniformly in both
    distributions,
    and from the fact that $h$ is also sampled uniformly from a pairwise
    independent hash family.
    The second equality simply used the linearity of expectation
    and the definition of expectation. 
    The third equality again uses the definition of expectation.
    We now continue to manipulate the expression above as follows:

    \begin{align*}
            \Delta&(P,\wt{P}) 
            =\E_k\left[\sum_h\E_{x\randfrom \D}\left[\frac{1}{2}\sum_{x'}\abs{\frac{1}{2^{|h|}}\Pr_{x''\randfrom \D}[x'' = x' | h(x'') = h(x)] - \frac{1}{2^{|h|}}\Pr[\mathcal{O}(k,h,h(x)) = x']}\right]\right]\\
            \\&\geq \E_k\left[\E_{x\randfrom \D}\left[\frac{1}{2}\sum_{x'}\abs{\sum_h\left(\frac{1}{2^{|h|}}\Pr_{x''\randfrom \D}[x'' = x' | h(x'') = h(x)] - \frac{1}{2^{|h|}}\Pr[\mathcal{O}(k,h,h(x)) = x']\right)}\right]\right]\\
            \\&= \E_k\left[\E_{x\randfrom \D}\left[\frac{1}{2}\sum_{x'}\abs{\Pr_{h,x''\randfrom \D}[x'' = x' | h(x'') = h(x)] - \Pr_h[\mathcal{O}(k,h,h(x)) = x']}\right]\right]\\
            \\&= \E_{k,x\randfrom \D}\left[\Delta(P_{k,x},\wt{P}_{k,x})\right].
    \end{align*}

    We again briefly explain each of the equalities above.
    The first equality
    simply rewrites the previous one.
    The second line (and only inequality) 
    uses the triangle inequality.
    The third and fourth equations use the definition of expectation.

    Since $\Delta(P,\wt{P}) \leq n^{-c}$, the inquality above implies
    $\E_{k,x}[\Delta(P_{k,x},\wt{P}_{k,x})] \leq n^{-c}$. 
    We can then show that,
    with probability $1-n^{-d}$ over $x$, 
    it holds that $\Delta(P_{k,x},\wt{P}_{k,x}) \leq n^{-d}$
    for all $k$.
    Indeed, suppose not: then with
    probability $n^{-d}$ over $x$, there exists a $k$ such that
    $\Delta(P_{k,x},\wt{P}_{k,x}) > n^{-d}$. Thus, with probability
    $n^{-d}\cdot n^{-1}$ over $x$ and $k$, we obtain $\Delta(P_{i,x},\wt{P}_{i,x}) >
    n^{-d}$. But then $\E_{k,x}[\Delta(P_{i,x},\wt{P}_{i,x})] > n^{-2d-1}
    \geq n^{-c}$, and we reach a contradiction.
\end{proof}


\begin{proof}[{Proof of~\Cref{claim:correctness}}]

    Let $m = -\log p_x$.
    We first show that $\cala(x) \geq p_x$ with probability at least
    $1-n^{-c}/2$.
    Note that if the test (Line~\ref{line:test} of \Cref{alg:est} stating
    that $c(k) \geq \frac{3}{8}$) passes for some $k \leq
    \ceil{m}$, then 
    $\cala(x) \geq 2^{-k} \geq 2^{-(\ceil{m}
    - 1)} = 2^{\floor{\log p_x} + 1} \geq 2^{\log p_x} = p_x$.
    Thus,
    $$\Pr_{x \flws \cald_n}[\cala(x) \geq p_x] 
    \geq \Pr_{x \flws \cald_n}[\text{test passes for }k=\ceil{m}].$$
    But we know that,
    for infinitely many $n$, 
    the probability that $\mathcal{O}(k,h,h(x)) =
    x$ for $k = \ceil{m}$ is at least $\frac{9}{10} - n^{-d} \geq 0.98$
    by~\Cref{claim:approx} and 
    the assumption,
    supposing $n$ is sufficiently large.
    Thus, by the Chernoff bound,
    we
    have 
    $$\Pr_{x \flws \cald_n}[\cala(x) \geq p_x] \geq
    1-\exp\left(-O(t)\right) \geq 1-n^{-c}/2,$$
    for infinitely many $n$, since $t = \poly(n)$.

    We now show that $\cala(x) \leq 2n^{2c}p_x$ with probability at least
    $1-n^{-c}/2$.
    Note that, as long as the test fails for all $k \leq m -2c\log n - 2$,
    then $\cala(x) \leq 4n^{2c}p_x$. 
    Since $n^{-d} \leq \frac{1}{16}$ for large enough
    $n$, by the Chernoff bound, Claim~\ref{claim:approx}
    and the assumption,
    the
    probability that the test passes for any given $k \leq m -2c\log n -2$ is
    at most $\exp(-\Omega(t))$.
    Therefore,
    by the union bound, we have that
    $$\Pr[\cala(x) \geq 4n^{2c}p_x] \leq
    2n \cdot \exp\left(-\Omega(n)\right)
    \leq
    n^{-c}/2,$$ 
    for infinitely many $n$,
    using also the fact that $t = \poly(n)$ is a large enough polynomial.
\end{proof}

\section{Computing GapK with Probability Estimation}

In this section, we prove the following theorem.

\begin{theorem}[Item 3 $\implies$ Item 2]
    \label{thm:compute-gapk}
    %
    %
    %
    %
    Suppose that probability estimation is 
    not weakly hard on average 
    on a quantum samplable distribution $\cald$.
    Then, for every $s : \bbn \to \bbn$ and $\Delta = \omega(\log n)$,
    the promise problem $\GapK[s,s+\Delta]$
    is 
    not
    weakly 
    hard on average 
    for quantum algorithms on
    $\cald$.
\end{theorem}

\subsection{Coding Theorem for Quantum Samplable Distributions}

First we need to show that probability estimation 
upper bounds
Kolmogorov complexity. 
To do this we need to generalize the coding theorem
to cover quantum samplable distributions. 
This can be done because 
Kolmogorov complexity is a time-unbounded notion
and quantum algorithms can be simulated by a time-unbounded 
classical machine. 

\begin{theorem}[Coding Theorem for Quantum Samplable Distributions]
    \label{thm:coding}
    For any quantum samplable distribution $\D$, 
    and any $x \in \blt^n$,
    we have 
    $$\K(x) \leq -\log(\Pr[\D_n \rightarrow x]) + \abs{\cald} + O(\log n).$$
\end{theorem}

\begin{proof}

As shown in~\cite[Lemma 3.2]{fortnow1999complexity}, there exists a
(classical) algorithm that, given $x$, computes the probability $p_x$ 
that
$\cald_n$ outputs $x$.
Let $\rho$ be the code of this algorithm.
Observe that the length of $\rho$ is $\abs{\cald_n} + O(1) = \abs{\cald} +
O(\log n)$.
If we sort the strings in $\blt^n$
in decreasing order of probability,
the string $x$ will appear in the first $\ceil{1/p_x}$ elements of the list.
Thus, we can specify the index of $x$ in that list with
$\log(1/p_x) + O(1)$ bits.
Therefore, 
a universal Turing machine can recover $x$
given $\log(1/p_x) + \abs{\cald} + O(\log n)$ bits.
%
\end{proof}


\subsection{Low Complexity, High Probability Strings are Uncommon}

Via the same argument given in \cite{ilango2021hardness}, we can show that,
for any quantum samplable distribution, probability estimation on average
implies Kolmogorov complexity estimation on average. 

From the 
argument in the previous subsection,
we know that high probability outputs have low Kolmogorov complexity, and
assuming that low probability outputs have high Kolmogorov complexity can only
hurt us in the very low probability instances.
More formally, we show that, given oracle access to an algorithm that
performs probability estimation on average on a distribution $\cald$,
we can solve $\GapK$ on average on that same distribution.
\Cref{thm:compute-gapk} will follow from the lemma below.

%
%
%

\begin{lemma}
    \label{lem:compute-gapk-oracle}
    Let $\cald$ be a quantum samplable distribution and 
    suppose $\Delta = \omega(\log n)$.
    Let $\calo$ be an oracle, $c \geq 1$ be a constant and $\eps : \bbn \to (0,1)$ 
    be such that
    \begin{equation*}
        \Pr_{x \flws \cald_n}
        [p_x/n^{-c} \leq \calo(x) \leq p_x]
        \geq 1-\eps.
    \end{equation*}
    There exists a QPT algorithm which, given oracle access to $\calo$,
    solves $\GapK[s-\Delta,s]$ on $\cald$
    with error at most $\eps + 2^{-\Delta/3}$, for any choice of $s$.
\end{lemma}

\begin{proof}
    Let $\alpha := 2^{-s+\Delta/2}$.
    The algorithm simply queries $\calo$ and accepts 
    if
    $\calo(x) \geq \alpha$,
    and rejects when
    $\calo(x) < \alpha$.

    We first show that the error from mischaracterizing high Kolmogorov complexity
    outputs is zero. 
    From \Cref{thm:coding} above, we get 
    that,
    if $p_x \geq \alpha$,
    then
    \begin{equation*}
        \K(x)
        \leq
        s-\Delta/2 + O(\log n)
        =
        s-\omega(\log n).
    \end{equation*}
    Therefore, for large enough $n$, 
    no string $x$ with Kolmogorov complexity
    at least $s$ 
    satisfies $p_x \geq \alpha$,
    and the algorithm never errs when the oracle doesn't.

    Furthermore, we show that the error from mischaracterizing low
    Kolmogorov
    complexity outputs is small. 
    Supposing the oracle is correct,
    we only make a mistake on inputs $x$ such that
    $\K(x) \leq s-\Delta$
    and
    $p_x \leq n^c\alpha$.
    There are at most
    $2^{s-\Delta + 1}$ strings with Kolmogorov complexity less than $s-\Delta$. 
    Their
    total probability is 
    at most
    $$2^{s-\Delta + 1} \cdot n^c\alpha = 2n^c  \cdot 2^{s-\Delta - s +
    \Delta/2} = 2n^c \cdot 2^{-\Delta/2} \leq 2^{-\Delta/3}$$
    for sufficiently large $n$ since $\Omega = \omega(\log n)$.

    In conclusion, adding up the error of the oracle $\calo$, the total error of the
    algorithm is $\eps + 2^{-\Delta/3}$.
\end{proof}

\Cref{thm:compute-gapk} now follows from \Cref{lem:compute-gapk-oracle}
since, if probability estimation is easy on average on $\cald$, then
there exists for every $q \geq 1$ a quantum algorithm that satisfies the
assumption of the oracle of the lemma, with error at most $n^{-q}$.
Furthermore, the error $2^{-\Delta/3}$ is negligible because 
$\Delta = \omega(\log n)$.

\section{Breaking One-Way Puzzles with a GapK oracle}
In this section we complete the characterization by showing that being able
to estimate $\GapK$ on any quantum samplable distribution is sufficient for
breaking $\OWP$. In \cite{ilango2021hardness}, this direction is very simple
since one way functions are known to imply pseudo-random generators
\cite{HILL99}, which can be broken by estimating any meta-complexity
measure (such as $\GapK$) which distinguishes between random and non-random
strings. However, a quantum equivalent to \cite{HILL99} is not known and
neither $\OWP$ nor even $\OWSG$ are known to imply a quantum 
pseudorandom ``generator''. 

However, we observe that, implicit in the proof of Corollary 14
of~\cite{chung2024central},
is a statement in that direction which is enough for our purposes.
They show that, if one-way puzzles exist,
then there exists a non-uniform QPT sampling algorithm $\cald$
such that, for some advice, the distribution $\cald$ is indistinguishable from
uniform, and exhibits an entropy gap.
Since the entropy of the distribution is small, we can argue that its
Kolmogorov complexity is small as well.
Moreover, the sampler crucially uses only $O(\log n)$ bits of
non-uniformity,
which means that we can employ a $\GapK$ oracle that works on a
uniformly quantum samplable distribution 
to distinguish the sampler $\cald$ from the uniform distribution,
thus breaking its security.

\begin{theorem}[\protect{\cite[Proof of Corollary 14]{chung2024central}}]
    \label{thm:nuefidef}
    If one way puzzles exist, there exists a 
    polynomial-time quantum algorithm $\cald$
    with the following properties.
    The algorithm $\cald$
    takes in
    two inputs $1^n$ and $\nu$,
    and outputs $m(n) > n$ bits.
    Moreover, for 
    each
    sufficiently large $n$, 
    there exists a binary string $\nu^*(n)$
    such that
    \begin{enumerate}
        \item $\nu^*(n) \in \blt^{t(n)}$ and $t(n) = O(\log n)$;
        \item $\D_n(\nu^*(n))$ is computationally indistinguishable from the uniform distribution on $m(n)$ bits;
        \item $H(\D_n(\nu^*(n))) \leq m(n) - n$,
    \end{enumerate}
    where $\cald_n(\nu) := \cald(1^n, \nu)$.
\end{theorem}


\begin{lemma}\label{claim:highk}
    For all $c$, $m$ and $n$,
    we have
    $$\Pr[\K(\U_m) \leq m - c\log n] \leq n^{-c}.$$
\end{lemma}

\begin{proof}
    There are $2^{m-c\log n}$ Turing machines of length $\leq m-c\log n$,
    and so there are at most $2^{m-c\log n}$ strings with $\K(x) \leq
    m-c\log n$. 
    Therefore,
    \begin{equation*}
        \Pr[\K(\U_m) \leq m - c\log n] \leq \frac{2^{m-c\log n}}{2^m} =
        n^{-c}.
    \end{equation*}
\end{proof}

\begin{lemma}\label{claim:lowk}
    Let $\cald$ and $\nu^*(n)$ be as in \Cref{thm:nuefidef}.
    Then
    $$\Pr[\K(D_n(\nu^*(n))) \leq m(n) - n + O(\log n)] \geq \frac{1}{m}.$$
\end{lemma}

\begin{proof}
    In this proof we will write $\cald_n$ to denote
    $\cald_n(\nu^*(n))$ for simplicity.

    Let $p_x := \Pr[\cald_n \to x]$.
    Observe that $H(\D_n) = \E_{x\randfrom \D_n}[-\log p_x]$. 
    Since $H(\D_n) \leq m - n$, by Markov's bound we get
    $$\Pr_{x\randfrom \D_n}[-\log p_x \geq m - n + 1] \leq \frac{m - n}{m -
    n + 1} = 1 - \frac{1}{m - n + 1}.$$

    But note that $\D_n$ is samplable by a constant size Turing machine on
    input $n,\nu^*(n)$. Since $|n,\nu^*(n)| = O(\log n)$, by \Cref{thm:coding}
    and the previous inequality
    we obtain 
    $$
    \Pr_{x\randfrom \D_n} [\K(x) \leq m - n + O(\log n)] 
    \geq
    \Pr_{x\randfrom \D_n}[-\log p_x \leq m - n + 1] 
    \geq
    \frac{1}{m-n+1} 
    \geq 
    \frac{1}{m},
    $$
    and the claim follows.
\end{proof}

\begin{theorem}[Item 1 $\implies$ Item 3]
    If one-way puzzles exist, then there exists 
    $s = n^{\Omega(1)}$,
    $\Delta = \omega(\log n)$ 
    and
    a quantum samplable distribution
    $\mathcal{S}$ 
    such that $\GapK[s,s+\Delta]$ is weakly average-case
    hard on $\mathcal{S}$.
\end{theorem}

\begin{proof}
    Assume one-way puzzles exist.
    Let $\cald, m$ and $\nu^*$ be 
    as
    given by~\Cref{thm:nuefidef}.
    We define the following distribution $\mathcal{S}_n$:
    \begin{enumerate}
        \item Sample $b\randfrom \{0,1\}$
        \item If $b=0$, sample 
            $s \randfrom \blt^{t(n)}$. 
            Output $\D_n(s)$.
        \item If $b=1$, output $r \randfrom \U_{m}$.
    \end{enumerate}

    Suppose for contradiction that, for all samplable distributions, for all $s=n^{\Omega(1)}$
    and $\Delta = \omega(\log n)$, the problem $\GapK[s,s+\Delta]$ is not weakly
    average-case hard.
    Let $c$ be such that 
    $m, 2^{\abs{\nu^*}} \leq n^c$. 
    Note that, since $m = n^{O(1)}$,
    we have $n = m^{\Omega(1)} = \omega(\log m)$.
    Thus, we have
    $(m-2c\log n) - (m-n/2) = \omega(\log m)$.
    Let $\mathcal{O}$ solve
    $\GapK[m-n/2,m - 2c\log n]$ with success probability $\geq 1-\frac{1}{4n^{2c}}$
    over $\mathcal{S}_n$.

    We claim that $\mathcal{O}$ breaks the security of the sampler as in~\Cref{thm:nuefidef}. 
    Observe that $\mathcal{S}_n$ conditioned on $b=0$ and $s = \nu^*$ is
    exactly $\D_n(\nu^*)$. 
    Since this condition holds with probability $\geq \frac{1}{2n^c}$, we
    have that $\mathcal{O}$ correctly distinguishes Kolmogorov complexity
    over $\D_n(\nu^*)$ with success probability $\geq 1-\frac{2n^c}{4n^{2c}}
    \geq 1 - \frac{1}{2n^c}$. Thus, by~\Cref{claim:lowk}, 
    $$
    \Pr[\mathcal{O}(\D_n(\nu^*)) \to 1] 
    \geq \frac{1}{n^c} -
    \frac{1}{2n^c} \geq \frac{1}{2n^c}.$$
    Similarly, we know that $\mathcal{O}$ correctly distinguishes
    Kolmogorov complexity over $\U_m$ with success probability 
    $\geq 1 - \frac{2}{4n^{2c}} \geq 1-\frac{1}{2n^{2c}}$. By~\Cref{claim:highk}, we
    have 
    $$\Pr[\mathcal{O}(\U_m) \to 1] \leq \frac{1}{n^{2c}} +
    \frac{1}{2n^{2c}} = \frac{3}{2n^{2c}}.$$
    This means that $\cald_n(\nu^*)$ is not indistinguishable from uniform,
    which contradicts 
    the security of the sampler $\cald$.


\end{proof}

\section*{Acknowledgements}
We thank Yanyi Liu and Angelos Pelecanos for collaborating on an earlier
stage of this project.
We thank the Simons Institute for the Theory of Computing for hosting
the Meta-Complexity program where this collaboration began.
B. Cavalar acknowledges support of
Royal Society University Research Fellowship URF$\backslash$R1$\backslash$211106. Eli Goldin is supported by a NSF Graduate Student Research Fellowship.

\bibliographystyle{alpha}
\bibliography{refs/crypto,refs/abbrev0,refs/bibliography}

%
\appendix

\end{document}